\documentclass{elsarticle}
\usepackage{amsthm}
\usepackage{amsmath}
\usepackage{amssymb}
\usepackage[boxed,ruled,noend,noline,linesnumbered]{algorithm2e}
\usepackage{hyperref}
\usepackage{fixltx2e}
\usepackage{caption}
\usepackage[table]{xcolor}
\usepackage{chngcntr}
\usepackage{comment}

\mathchardef\mhyphen="2D

\newcommand\placeqed{\tiny{\nobreak\enspace$\square$}}
\newtheorem{mydef}{Definition}
\newtheorem{proposition}{Proposition}
\newcommand\myeq{\mathrel{\stackrel{\makebox[0pt]{\mbox{\normalfont\tiny def}}}{=}}}

\makeatletter
\def\ps@pprintTitle{%
	\let\@oddhead\@empty
	\let\@evenhead\@empty
	\let\@evenfoot\@oddfoot}
	
\makeatother

\begin{document}

\begin{frontmatter}

\title{On Computing Stable Extensions of Abstract Argumentation Frameworks}

\author{Samer Nofal}
\author{Amani Abu Jabal}
\author{Abdullah Alfarrarjeh}
\author{Ismail Hababeh}
\address{Department of Computer Science, German Jordanian University, Jordan}

\begin{abstract}
An \textit{abstract argumentation framework} ({\sc af} for short) is a directed graph $(A,R)$ where $A$ is a set of \textit{abstract arguments} and $R\subseteq A \times A$ is the \textit{attack} relation. Let $H=(A,R)$ be an {\sc af}, $S \subseteq A$ be a set of arguments and $S^+ = \{y \mid \exists x\in S \text{ with }(x,y)\in R\}$. Then, $S$ is a \textit{stable extension} in $H$ if and only if $S^+ = A\setminus S$. In this paper, we present a thorough, formal validation of a known backtracking algorithm for listing all stable extensions in a given {\sc af}.
\end{abstract}


\end{frontmatter}

\section{Introduction}

 An \textit{abstract argumentation framework} ({\sc af} for short) is a directed graph $(A,R)$ where $A$ is a set of \textit{abstract arguments} and $R \subseteq A \times A$ is the \textit{attack} relation between them. Let $H=(A,R)$ be an {\sc af}, $S \subseteq A$ be a subset of arguments and $S^+ = \{y \mid \exists x\in S \text{ with }(x,y)\in R\}$. Then, $S$ is a \textit{stable extension} in $H$ if and only if $S^+=A\setminus S$.

Since introduced in \cite{DBLP:journals/ai/Dung95}, {\sc af}s have attracted a substantial body of research (see e.g. \cite{DBLP:journals/aim/AtkinsonBGHPRST17,TheAddedValueOfArgumentation,DBLP:journals/ker/BaroniCG11}). Stable extension enumeration is a fundamental computational problem within the context of {\sc af}s. Listing all stable extensions in an {\sc af} is an {\sc np}-hard problem, see for example the complexity results presented in \cite{DUNNE2007701,DBLP:journals/flap/DvorakD17}. In the literature one can find different proposed methods for listing all stable extensions in a given {\sc af}, such as dynamic programming and reduction-based methods, see for example \cite{DBLP:journals/ai/CharwatDGWW15,DBLP:journals/flap/CeruttiGTW17} for a fuller review. However, this paper is centered around backtracking algorithms for listing stable extensions in an {\sc af}. 

In related work, backtracking algorithms are often called \textit{labelling} algorithms. The work of \cite{DBLP:journals/amai/DimopoulosMP97} proposed a backtracking algorithm that can be used for generating stable extensions of {\sc af}s. Later, the work of \cite{DBLP:journals/ijar/NofalAD16} enhanced the algorithm of \cite{DBLP:journals/amai/DimopoulosMP97} by a look-ahead mechanism. Most recently, the system of \cite{DBLP:conf/tafa/GeilenT17} implemented a backtracking algorithm that is similar to the essence of the algorithm of \cite{DBLP:journals/ijar/NofalAD16}; however, \cite{DBLP:conf/tafa/GeilenT17} made use of heuristics to examine their effects on the practical efficiency of listing stable extensions. All these works (i.e. \cite{DBLP:journals/amai/DimopoulosMP97,DBLP:journals/ijar/NofalAD16,DBLP:conf/tafa/GeilenT17}) presented \textit{experimental} backtracking algorithms. 

Generally, there are several empirical studies on algorithmic methods in abstract argumentation research, see e.g. \cite{DBLP:journals/aim/ThimmVCOSV16,AIMAGcompetition,DBLP:journals/logcom/BistarelliRS18}. However, to the best of our knowledge, we do not find in the literature a comprehensive, formal demonstration of a backtracking algorithm for listing stable extensions in a given {\sc af}. Therefore, in the present paper we give a rigorous, formal validation of a backtracking algorithm (which is comparable to the core structure of the algorithm of \cite{DBLP:journals/ijar/NofalAD16}) for listing all stable extensions in an {\sc af}. 

In section 2, we validate a first version of the algorithm, which is a formulation that naturally builds on the definition of stable extensions. In section~3, we prove a second version of the algorithm, which is a development of the first version. As the first version of the algorithm subsumes (somewhat) high-level operations, the second version comes to specify how these operations are implemented efficiently. We believe that proving first the algorithm in a high-level form enables the reader to follow easily the validation of the implementation given in the second version. The second version is the ultimate construct that one needs to realize the algorithm in any prospective application. In section 4, we close the paper with some concluding remarks.

\section{Validation of an algorithm for listing stable extensions}
Let $H=(A,R)$ be an {\sc af} and $T \subseteq A$ be a subset of arguments. Then,
\begin{equation*}
\begin{array}{l}
	T^+ \myeq \{y \mid \exists x\in T \text{ with }(x,y)\in R\},\\
	T^- \myeq \{y \mid \exists x\in T \text{ with }(y,x)\in R\}. 
	\end{array}
\end{equation*}	
	For the purpose of enumerating all stable extensions in $H$, let $S$ denote an under-construction stable extension of $H$. Thus, we start with $S=\emptyset$ and then let $S$ grow to a stable extension (if any exists) incrementally by choosing arguments from $A$ to join $S$. For this, we denote by $choice$ a set of arguments eligible to join $S$. More precisely, take $S\subseteq A$ such that $S \cap S^+=\emptyset$, then $choice \subseteq A \setminus (S \cup S^- \cup S^+)$. Additionally, we denote by $tabu$ the arguments that do not belong to $S \cup S^+ \cup choice$. More specifically, take $S\subseteq A$ such that $S\cap S^+=\emptyset$, and $choice\subseteq A \setminus (S \cup S^- \cup S^+)$. Then, $tabu= A \setminus (S \cup S^+ \cup choice)$. The next example shows these structures in action.  \\

\noindent \textbf{Example 1.} To list the stable extensions in the {\sc af} $H_{1}$ depicted in Figure~1, apply the following steps:
 \begin{enumerate}
  \item Initially, $S=\emptyset$, $S^+=\emptyset$, $choice =\{a,b,c,d,e,f\}$, and $tabu=\emptyset$. 
  \item Select an argument, say $a$, to join $S$. Then, $S=\{a\}$, $S^+=\{b\}$, $choice=\{c,d\}$ and $tabu=\{e,f\}$. 
  \item As $\{c,d\}^- \subseteq S^+ \cup tabu$, $\{c,d\}$ must join $S$. Otherwise, $S=\{a\}$ without $\{c,d\}$ will never expand to a stable extension since $c$ and $d$ never join $S^+$. Thus, $S=\{a,c,d\}$, $S^+=\{b,e,f\}$, $choice=\emptyset$ and $tabu=\emptyset$. As $S^+=A \setminus S$, $S$ is stable. 
  \item To find another stable extension, backtrack to the state of step 1 and then try to build a stable extension without~$a$. Hence, $S=\emptyset$, $S^+=\emptyset$, $choice=\{b,c,d,e,f\}$ and $tabu=\{a\}$. 
  \item Select an argument, say $b$, to join $S$. Then, $S=\{b\}$, $S^+=\{c,d\}$, $choice=\{e,f\}$ and $tabu=\{a\}$. 
  \item As $\{e\}^- \subseteq S^+$, $e$ must join $S$. Otherwise, $S=\{b\}$ without $e$ will not grow to a stable extension since $e$ never joins $S^+$. Hence, $S=\{b,e\}$, $S^+=\{a,c,d,f\}$, $choice=\emptyset$ and $tabu=\emptyset$. Since $S^+= A\setminus S$, $S$ is stable. 
  \item Backtrack to the state of step 4 and then attempt to build a stable extension excluding $b$. Thus, $S=\emptyset$, $S^+=\emptyset$, $choice=\{c,d,e,f\}$ and $tabu=\{a,b\}$. 
  \item Since $\{d\}^-\subseteq tabu$, $d$ must join $S$. Otherwise, $S=\emptyset$ without $d$ will never grow to a stable extension because $d$ never joins $S^+$. Subsequently, $S=\{d\}$, $S^+=\{b,e,f\}$, $choice=\{c\}$ and $tabu=\{a\}$.
  \item As $\{a\}^-\subseteq S^+$ and $a \in tabu$, $a$ will never join $S^+$ and so $S=\{d\}$ will never grow to a stable extension. 
\item At this point, we confirm that there are no more stable extensions to find. This is because the state of step 7, being analyzed in steps 8 \&~9, an assertion is concluded that trying to build a stable extension from $choice=\{c,d,e,f\}$ (i.e. excluding $tabu=\{a,b\}$) will never be successful. Note, we already tried to build a stable extension including $a$ (step 2), and later (step 5) we tried to construct a stable extension including $b$ but without $a$.
\end{enumerate}

\begin{figure}
	\centering	
	\includegraphics[width=5cm]{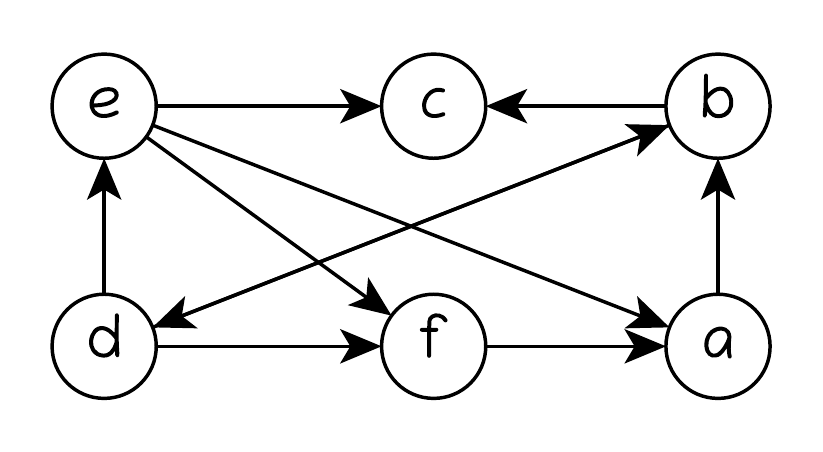}
	\caption{Argumentation framework $H_{1}$.}
\end{figure}

Our first proposition captures stable extensions when $choice=tabu=~\emptyset$. 

\begin{proposition}
	Let $H=(A,R)$ be an {\sc af}, $S\subseteq A$, $S \cap S^+ = \emptyset$, $choice \subseteq A\setminus (S \cup S^+\cup S^-)$, $tabu = A\setminus (S\cup S^+\cup choice)$. Then, $S$ is a stable extension in $H$ if and only if $choice=\emptyset$ and $tabu=\emptyset$. 
\end{proposition}
\begin{proof}	
	\renewcommand\qedsymbol{}
	If $tabu=\emptyset$ and $choice=\emptyset$, then $\emptyset = A\setminus (S\cup S^+\cup \emptyset)$ because $tabu = A\setminus (S\cup S^+\cup choice)$. Thus, $S\cup S^+ = A$. As $S \cap S^+ = \emptyset$, then $S^+ = A\setminus S$. On the other hand, if $S$ is stable, then $S^+=A \setminus S$ and so $S \cup S^+ =A$. Thus, $tabu= A \setminus (S \cup S^+ \cup choice)=A\setminus (A\cup choice)=\emptyset$. Further, $choice=\emptyset$ since $choice \subseteq A \setminus (S\cup S^+ \cup S^-)$ and $S \cup S^+=A$.\placeqed
\end{proof}
 
 Now, we present three propositions that are essential for efficiently listing stable extensions. These propositions are inspired by the excellent work of \cite{DBLP:conf/cade/DoutreM01}, which presented a backtracking algorithm for solving a different computational problem in {\sc af}s.
 \numberwithin{equation}{proposition}
\begin{proposition}
 	Let $H=(A,R)$ be an {\sc af}, $S\subseteq A$, $S \cap S^+ = \emptyset$, $choice \subseteq A\setminus (S \cup S^+\cup S^-)$, $tabu = A\setminus (S\cup S^+\cup choice)$, and $x \in choice$ be an argument with $\{x\}^- \subseteq S^+ \cup tabu$. If there is a stable extension $T\supseteq S$ such that $T \setminus S \subseteq choice$, then $x \in T$.	
\end{proposition}	
\begin{proof}
	\renewcommand\qedsymbol{}
	Suppose $x \notin T$. Then, $x \in T^+$ because $T$ is stable. Thus, 
	\begin{equation}
	\exists y \in \{x\}^-\text{ such that }y \in T.
	\end{equation}
As $\{x\}^- \subseteq S^+ \cup tabu$,
	\begin{equation}
	\forall y \in \{x\}^-\text{ it is the case that }y \in S^+ \cup tabu.
	\end{equation}
Hence, (2.1) together with (2.2) imply that
\begin{equation}
T \cap (S^+ \cup tabu) \neq \emptyset.
\end{equation}
Observe, $T \subseteq S \cup choice$ since $T \supseteq S$ and $T\setminus S \subseteq choice$. As $S \cap (S^+\cup tabu)=\emptyset$ and $choice \cap (S^+\cup tabu)=\emptyset$, it holds that 
\begin{equation}
T \cap (S^+ \cup tabu) =\emptyset.
\end{equation}
Note the contradiction between (2.3) and (2.4). \placeqed
\end{proof}
\begin{proposition}
	Let $H=(A,R)$ be an {\sc af}, $S\subseteq A$, $S \cap S^+ = \emptyset$, $choice \subseteq A\setminus (S \cup S^+\cup S^-)$, $tabu = A\setminus (S\cup S^+\cup choice)$, and $x \in choice$ be an argument such that for some $y \in tabu$ it is the case that $\{y\}^- \cap choice=\{x\}$. If there is a stable extension $T\supseteq S$ such that $T \setminus S \subseteq choice$, then $x \in T$. 	
\end{proposition}
\begin{proof}\renewcommand\qedsymbol{}
	Suppose $x \notin T$. As $x \in choice$,
	\begin{equation}
	T \setminus S \subseteq choice \setminus \{x\}. 
	\end{equation}
As $T\supseteq S$ and $T\setminus S \subseteq choice$, it is the case that $T\subseteq S \cup choice$. Because $tabu \cap (S\cup choice)=\emptyset$, it holds that
\begin{equation}
T \cap tabu=\emptyset.
\end{equation}
Referring to the premise of this proposition, as $y \in tabu$, $y \notin T$. Subsequently, $y \in T^+$ since $T$ is stable. Thus, 
\begin{equation}
\{y\}^- \cap T \not = \emptyset.
\end{equation}
Since $y \in tabu$ and $tabu \cap S^+=\emptyset$, $y \notin S^+$. Thus, \begin{equation}\{y\}^- \cap S = \emptyset.\end{equation} 
Due to (3.1), (3.3), and (3.4), it holds that $\{y\}^- \cap (choice \setminus \{x\}) \not =\emptyset$, which is a contradiction with $\{y\}^- \cap choice= \{x\}$, see the premise of this proposition.\placeqed
\end{proof}
\begin{proposition}
	Let $H=(A,R)$ be an {\sc af}, $S\subseteq A$, $S \cap S^+=\emptyset$, $choice \subseteq A\setminus (S \cup S^+\cup S^-)$, $tabu = A\setminus (S\cup S^+\cup choice)$, and $x \in tabu$ be an argument with $\{x\}^- \subseteq S^+ \cup tabu$. Then, there does not exist a stable extension $T \supseteq S$ such that $T\setminus S \subseteq choice$. 	
\end{proposition}	
\begin{proof}\renewcommand\qedsymbol{}
	Assume there exists a stable extension $T \supseteq S$ with $T\setminus S \subseteq choice$. Since $(S\cup choice) \cap (S^+ \cup tabu) =\emptyset$, 
	\begin{equation}
	T \cap (S^+ \cup tabu)=\emptyset.
	\end{equation}
	As $x \in tabu$ and due to (4.1), $x \notin T$. Thus, $x \in T^+$ because $T$ is stable. Subsequently, 
	\begin{equation}
	T \cap \{x\}^- \neq \emptyset.
	\end{equation}
	Because $\{x\}^- \subseteq S^+ \cup tabu$,
	\begin{equation}
T \cap (S^+ \cup tabu) \neq \emptyset
	\end{equation}
	 Note the contradiction between (4.1) and (4.3).\placeqed
\end{proof}
We now give Algorithm 1 for listing all stable extensions in a given {\sc af}. Before presenting its proof, we demonstrate an execution of Algorithm~1.

\begin{algorithm}
	\caption{stb($S$, $choice$, $tabu$)}
	\Repeat{$\alpha = \emptyset$ and $\beta =\emptyset$}{
		\textbf{if} there is $x \in tabu$ with $\{x\}^- \subseteq S^+ \cup tabu$ \textbf{then} return\;	
		$\alpha \gets \{x \in choice \mid \{x\}^- \subseteq S^+\cup tabu\}$\;		
		$S \gets S\cup \alpha$\;
		$choice \gets choice \setminus (\alpha \cup \alpha^+ \cup \alpha ^-)$\;
		$tabu \gets (tabu \cup \alpha^-) \setminus S^+$\;
		\textbf{if} $\neg (\exists y \in choice, \exists x \in tabu, \{x\}^- \cap choice=\{y\})$ \textbf{then} $\beta \gets \emptyset$ \textbf{else}\\
		\hspace{12pt}$\beta \gets \{y\}$ \;
		\hspace{12pt}$S \gets S\cup \beta$\;
		\hspace{12pt}$choice \gets choice \setminus (\beta \cup \beta^+ \cup \beta^-)$\;
		\hspace{12pt}$tabu \gets (tabu \cup \beta^-) \setminus S^+$\;
		
	}	
	\If{$choice=\emptyset$}{
		\textbf{if} $tabu=\emptyset$ \textbf{then} $S$ is stable\;	
		return\;		
	}
	\tcp{For some $x \in choice$.}
	stb($S\cup \{x\}$, $choice\setminus (\{x\}\cup \{x\}^+\cup \{x\}^-)$, $(tabu \cup \{x\}^-)\setminus (S^+\cup \{x\}^+)$)\;
	stb($S$, $choice \setminus \{x\}$, $tabu \cup \{x\}$)\;
	
\end{algorithm} 
\vspace{3mm}
\noindent \textbf{Example 2.} Apply Algorithm 1 to list the stable extensions in $H_1$:
\begin{enumerate}
	\item Initially, call $stb(\emptyset,\{a,b,c,d,e,f\},\emptyset)$.
	\item Perform the repeat-until block of Algorithm 1 to get $S=\emptyset$, $choice=\{a,b,c,d,e,f\}$, and $tabu=~\emptyset$.
	\item Call $stb(\{a\},\{c,d\},\{e,f\})$, see line 16 in the algorithm.
	\item Now, execute the repeat-until block to get $S=\{a,c,d\}$, $choice=\emptyset$, and $tabu=~\emptyset$. Applying line 14, $S$ is stable. Now, apply line 15 and return to the state: $S=\emptyset$, $choice=\{a,b,c,d,e,f\}$, and $tabu=~\emptyset$.
	\item Call $stb(\emptyset,\{b,c,d,e,f\},\{a\})$, see line 17 in the algorithm.
	\item After performing the repeat-until block, $S=\emptyset$, $choice=\{b,c,d,e,f\}$, and $tabu=\{a\}$.
	\item Call $stb(\{b\},\{e,f\},\{a\})$, see line 16.
	\item Apply the repeat-until block to get $S=\{b,e\}$, $choice=\emptyset$, and  $tabu=\emptyset$. Performing line 14, we find $S$ stable. Now apply line 15, and so return to the state: $S=\emptyset$, $choice=\{b,c,d,e,f\}$, and $tabu=\{a\}$.  
	\item Call $stb(\emptyset,\{c,d,e,f\},\{a,b\})$, see line 17.
	\item Perform a first round of the repeat-until block. Thus, $S=\{d\}$, $choice=\{c\}$, and $tabu=\{a\}$. In a second round of the repeat-until block, the algorithm returns (see line 2), and eventually the algorithm halts.	
\end{enumerate}

In what follows we give another set of propositions that together with the previous ones will establish the validity of Algorithm~1. Thus, we denote by $T_i$ the elements of a set $T$ at the algorithm's state $i$. The algorithm enters a new state whenever $S$ is updated. In other words,  the algorithm enters a new state whenever line 4, 9, 16, or 17 are executed. Focusing on the under-construction stable extension $S$ and consistently with the algorithm's actions, in the initial state of the algorithm we let
\setcounter{equation}{0}
\counterwithout{equation}{proposition}
\begin{equation}
S_{1}=\emptyset,
\end{equation} 
and for every state $i$ it is the case that
\begin{equation}
\begin{array}{l}
S_{i+1}=S_{i}\text{, or }\\
S_{i+1}= S_{i} \cup \delta_{i} \text{ with } \delta_{i} \in \{\alpha_{i},\beta_{i},\{x\}_{i}\};
\end{array}
\end{equation}
see respectively lines 17, 4, 9, and 16 in the algorithm.

 \numberwithin{equation}{proposition}
\begin{proposition}
	Let $H=(A,R)$ be an {\sc af} and Algorithm 1 be started with \begin{center}
		stb$(\emptyset, A \setminus \{x\mid (x,x)\in R\}, \{x\mid (x,x)\in R\})$.
	\end{center} 
For every state $i$, $S_i \cap S_i^+ = \emptyset$.
\end{proposition}
\begin{proof}\renewcommand\qedsymbol{}
	As $S_1=\emptyset$, $S_{1} \cap S_{1}^+=\emptyset$. We now show that for every state $i$
	\begin{equation}
	S_{i} \cap S_{i}^+=\emptyset \implies S_{i+1} \cap S_{i+1}^+=\emptyset.
	\end{equation}
	Suppose $S_{i} \cap S_{i}^+=\emptyset$. Using (2) we write
	\begin{equation}
	\begin{array}{ll}
	S_{i+1} \cap S_{i+1}^+&= (S_{i} \cup \delta_{i}) \cap (S_{i}^+ \cup \delta_{i}^{+})\\
	&=(S_i \cap S_i^+) \cup (S_i\cap \delta_{i}^+) \cup (S_i^+ \cap \delta_{i})  \cup (\delta_{i} \cap \delta_{i}^+).
	\end{array}
	\end{equation}	
	 In fact, $\delta_{i} \subseteq choice_i$, recall (2) and lines 3, 7, and 16. Considering (5.2), we proceed the proof by showing that
	 \begin{equation}	 
	 (S_i \cap S_i^+) \cup (S_i\cap choice_{i}^+) \cup (S_i^+ \cap choice_{i})  \cup (\delta_{i} \cap \delta_{i}^+) = \emptyset.	
	 \end{equation}
	 In other words, we need to prove that for every state $i$
	 \begin{equation}
	 S_i\cap choice_{i}^+=\emptyset,
	 \end{equation}
	\begin{equation}
	S_i^+ \cap choice_i =\emptyset,
	\end{equation}
	\begin{equation}
	\delta_{i} \cap \delta_{i}^+ = \emptyset.
	\end{equation}
	Now we prove (5.4). For $i=1$, $S_1 \cap choice_1^+=\emptyset$ since $S_1=\emptyset$. Then, we will show that
	\begin{equation}
	\forall i~ (S_i \cap choice_i^+ =\emptyset \implies S_{i+1} \cap choice_{i+1}^+= \emptyset).
	\end{equation}	
	Suppose $S_i \cap choice_i^+ = \emptyset$.  Referring to lines 5, 10 and 16, it holds that $choice_{i+1}=choice_i \setminus (\delta_i\cup \delta_i^- \cup \delta_i^+)$. And using (2), we note that
	\begin{equation}
		 S_{i+1} \cap choice_{i+1}^+=(S_i \cup \delta_i) \cap (choice_i \setminus (\delta_i \cup \delta_i^- \cup \delta_i^+))^+.
	\end{equation}
	Since $S_i \cap choice_i^+ = \emptyset$, it holds that $S_i  \cap (choice_i \setminus (\delta_i \cup \delta_i^- \cup \delta_i^+))^+=\emptyset$.
	Additionally, we note that
	\begin{equation}
	\delta_i \cap (choice_i \setminus (\delta_i \cup \delta_i^- \cup \delta_i^+))^+ =\emptyset \iff \delta_i^- \cap (choice_i \setminus (\delta_i \cup \delta_i^- \cup \delta_i^+)) =\emptyset.
	\end{equation}
	As  $\delta_i^- \cap (choice_i \setminus (\delta_i \cup \delta_i^- \cup \delta_i^+)) =\emptyset$, (5.8) holds. For the case of line 17 in the algorithm, where $S_{i+1}=S_i$, note that
	\begin{equation}
	S_{i+1} \cap choice_{i+1}^+  = S_i \cap (choice_i \setminus \{x\}_{i})^+  =\emptyset
	\end{equation}
	since $S_i \cap choice_i^+=\emptyset$. That concludes the proof of (5.4). Now we prove (5.5). For $i=1$, it holds that $choice_1 \cap S_1^+ =\emptyset$ since $S_1=\emptyset$. Then, we need to show that for every state $i$
	\begin{equation}
	choice_i \cap S_i^+=\emptyset \implies choice_{i+1} \cap S_{i+1}^+ =\emptyset.
	\end{equation}
	Referring to line 5, 10, and 16, it is the case that $choice_{i+1} = choice_i \setminus (\delta_{i} \cup \delta_{i}^+ \cup \delta_{i}^-)$. Considering (2), observe that
	\begin{equation}
	\begin{array}{ll}
	choice_{i+1} \cap S_{i+1}^+ &= (choice_i \setminus (\delta_{i} \cup \delta_{i}^+ \cup \delta_{i}^-)) \cap (S_i\cup \delta_{i})^+\\
	&=(choice_i \setminus (\delta_{i} \cup \delta_{i}^+ \cup \delta_{i}^-)) \cap (S_i^+\cup \delta_{i}^+).
	\end{array}
	\end{equation}
	Thus, (5.11) holds. Likewise, for the case of line 17 in the algorithm, 
	\begin{equation}
	choice_{i+1} \cap S_{i+1}^+ =(choice_i \setminus \{x\}_{i} ) \cap S_i^+ =\emptyset
	\end{equation}
	since $choice_i \cap S_i^+=\emptyset$, see the premise of (5.11).
	That completes the proof of (5.5). Now we prove (5.6). From (2), for all $i$, if $S_{i+1}=S_i \cup \delta_i$, then it holds that $\delta_{i} \in \{\alpha_{i},\beta_{i},\{x\}_{i}\}$. Considering the algorithm's actions (lines 3, 7 and 16), $\delta_i \subseteq choice_i$. If $\delta_i \in \{\beta_i,\{x\}_i\}$, then $|\delta_i|=1$, and so $\delta_i \cap \delta_i^+=\emptyset$ holds since $\{x\mid (x,x) \in R\} \not \subseteq choice_i$. If $\delta_i=\alpha_i$, then from line 3 in the algorithm we note that 
	\begin{equation}
	\delta_i = \{x \in choice_i \mid \{x\}^- \subseteq tabu_i \cup S_i^+\}.
	\end{equation}
	Therefore,
	\begin{equation}
	\delta_i^- \subseteq tabu_i \cup S_i^+.
	\end{equation}	
	To establish $\delta_i \cap \delta_i^+ = \emptyset$ it suffices to show that
	 \begin{equation}
	 \delta_i \cap \delta_i^- = \emptyset.
	 \end{equation} 
	Thus, given (5.15) with the fact that $\delta_i \subseteq choice_i $, we will prove (5.16) by showing that
	\begin{equation}
	choice_i \cap (tabu_i \cup S_i^+) = \emptyset.
	\end{equation}
	Because we already showed that $choice_i \cap S_i^+ = \emptyset$, recall (5.5), the focus now is on demonstrating that for all $i$
	\begin{equation}
	choice_i \cap tabu_i = \emptyset.
	\end{equation}
	For $i=1$, (5.18) holds since $choice_1 =A \setminus \{x\mid (x,x) \in R\}$ and $tabu_1 = \{x \mid (x,x) \in R\}$. Now we will show that 
	\begin{equation}
	\forall i~(choice_i \cap tabu_i=\emptyset \implies choice_{i+1} \cap tabu_{i+1}=\emptyset).
	\end{equation}
	Let $choice_i \cap tabu_i=\emptyset$. Observe that
	\begin{equation}
	\begin{array}{ll}
	choice_{i+1} = choice_i \setminus (\delta_i\cup \delta_i^-\cup \delta_i^+) & \text{ (lines 5, 10, and 16)}\\ tabu_{i+1}=(tabu_i\cup \delta_i^-)\setminus (S_i^+\cup \delta_i^+) & \text{ (lines 6, 11, and 16)}
	\end{array}
	\end{equation}
	Thus,
	\begin{equation}
	choice_{i+1} \cap tabu_{i+1}= (choice_i \setminus (\delta_i\cup \delta_i^-\cup \delta_i^+)) \cap ((tabu_i\cup \delta_i^-)\setminus (S_i^+\cup \delta_i^+))=\emptyset.
	\end{equation}
	For the case of line 17, observe that
	\begin{equation}
	choice_{i+1} \cap tabu_{i+1} = (choice_i \setminus \{x\}_i) \cap (tabu_i \cup \{x\}_i)=\emptyset.
	\end{equation}
	That concludes the demonstration of (5.6), and so the proof of (5.1) is now complete.\placeqed
\end{proof}
\begin{proposition}
	Let $H=(A,R)$ be an {\sc af} and Algorithm 1 be started with \begin{center}
		stb$(\emptyset, A \setminus \{x\mid (x,x)\in R\}, \{x\mid (x,x)\in R\})$.
	\end{center} 
For every state $i$, $choice_i \subseteq A \setminus (S_i \cup S_i^+ \cup S_i^-)$.
\end{proposition}
\begin{proof}\renewcommand\qedsymbol{}
	 We note that $choice_1 \subseteq A \setminus (S_1 \cup S_1^+ \cup S_1^-)$. This is because $choice_1=A \setminus \{x \mid (x,x) \in R\}$ and, $S_1=\emptyset$. Now, we show that for every state $i$ 
	\begin{equation}
	choice_i \subseteq A \setminus (S_i \cup S_i^+ \cup S_i^-) \implies choice_{i+1} \subseteq A \setminus (S_{i+1} \cup S_{i+1}^+ \cup S_{i+1}^-).
	\end{equation}
	Given the premise of (6.1), we write
	\begin{equation}
	choice_i \setminus (\delta_{i} \cup \delta_{i}^- \cup \delta_{i}^+) \subseteq A \setminus (S_i \cup S_i^+ \cup S_i^- \cup \delta_{i} \cup \delta_{i}^- \cup \delta_{i}^+).
	\end{equation}
	According to lines 5, 10 and 16 in the algorithm we observe that
	\begin{equation}
	choice_{i+1} = choice_i \setminus (\delta_{i} \cup \delta_{i}^+ \cup \delta_{i}^-).
	\end{equation}	
	Hence, using (6.3) along with (2) we rewrite (6.2) as
	\begin{equation}
	choice_{i+1} \subseteq A \setminus (S_{i+1} \cup S_{i+1}^+ \cup S_{i+1}^-),
	\end{equation}
	which is the consequence of (6.1). But to complete the proof of (6.1) we need to illustrate the case of line 17 in the algorithm where for some $x \in choice_i$ it holds that
	\begin{equation}
	choice_{i+1} = choice_i \setminus \{x\}_i,
	\end{equation}
	and
	\begin{equation}
	S_{i+1} = S_i.
	\end{equation}
Given the premise of (6.1), observe that
	\begin{equation}
	choice_i \setminus \{x\}_i \subseteq A \setminus (S_i \cup S_i^+ \cup S_i^-).
	\end{equation}
	Using (6.5) with (6.6), we note that $choice_{i+1} \subseteq A \setminus (S_{i+1} \cup S_{i+1}^+ \cup S_{i+1}^-)$.\placeqed	 
\end{proof}
\begin{proposition}
	Let $H=(A,R)$ be an {\sc af} and Algorithm 1 be started with \begin{center}
		stb$(\emptyset, A \setminus \{x\mid (x,x)\in R\}, \{x\mid (x,x)\in R\})$.
	\end{center} 
For every state $i$, $tabu_i = A \setminus (S_i \cup S_i^+ \cup choice_i)$.
\end{proposition}
\begin{proof}\renewcommand\qedsymbol{}
	We note that $tabu_1=A\setminus (S_1 \cup S_1^+ \cup choice_1).$ Observe, $S_1=\emptyset$, $choice_1=A\setminus \{x\mid (x,x) \in R\}$ and $tabu_1=\{x\mid (x,x) \in R\}$. Now we prove that for every state $i$,
	\begin{equation}
	tabu_i=A\setminus (S_i \cup S_i^+ \cup choice_i) \implies tabu_{i+1}=A\setminus (S_{i+1} \cup S_{i+1}^+ \cup choice_{i+1}).
	\end{equation}
	Considering lines 6, 11, and 16 we note that
	\begin{equation}
	tabu_{i+1}=(tabu_i \cup \delta_{i}^-) \setminus (S_i^+ \cup \delta_{i}^+).
	\end{equation}
	Given the premise of (7.1), rewrite (7.2) as
	\begin{equation}
	tabu_{i+1}=((A\setminus (S_i \cup S_i^+ \cup choice_i)) \cup \delta_{i}^-) \setminus (S_{i}^+ \cup \delta_{i}^+).\\	
	\end{equation}
	According to lines 5, 10 and 16 in the algorithm we observe that
	\begin{equation}
	choice_{i+1} = choice_i \setminus (\delta_{i} \cup \delta_{i}^+ \cup \delta_{i}^-).
	\end{equation}
	Note that $\delta_i \subseteq choice_i$, see (2). For the case where $\delta_i^+ \cup \delta_i^-  \subseteq choice_i$, using (7.4) and (2), we rewrite (7.3) as
	\begin{equation}
	\begin{array}{ll}
	tabu_{i+1}&=((A\setminus (S_i \cup S_i^+ \cup choice_{i+1} \cup \delta_{i} \cup \delta_{i}^+ \cup \delta_{i}^-)) \cup \delta_{i}^-) \setminus (S_{i}^+ \cup \delta_{i}^+)\\
	&=((A\setminus (S_{i+1} \cup S_{i+1}^+ \cup choice_{i+1} \cup \delta_{i}^-)) \setminus (S_{i}^+ \cup \delta_{i}^+)) \cup (\delta_{i}^- \setminus (S_{i}^+ \cup \delta_{i}^+))\\
	&=((A\setminus (S_{i+1} \cup S_{i+1}^+ \cup choice_{i+1} \cup \delta_{i}^-)) \setminus S_{i+1}^+ )\cup (\delta_{i}^- \setminus S_{i+1}^+ )\\
	&=(A\setminus (S_{i+1} \cup S_{i+1}^+ \cup choice_{i+1} \cup \delta_{i}^-)) \cup (\delta_{i}^- \setminus S_{i+1}^+ )
	\end{array}
	\end{equation}
	Note, it is always the case that $\delta_i^- \cap (choice_{i+1}\cup S_{i+1})=\emptyset$, see Proposition 5 and (7.4). Thus, (7.5) can be rewritten as
	\begin{equation}
	\begin{array}{ll}
tabu_{i+1}&=(A\setminus (S_{i+1} \cup S_{i+1}^+ \cup choice_{i+1} \cup (\delta_i^- \setminus S_{i+1}^+ )))\cup (\delta_i^- \setminus S_{i+1}^+ )\\
&=A\setminus (S_{i+1} \cup S_{i+1}^+ \cup choice_{i+1}).
	\end{array}	
	\end{equation}
For the case where $(\delta_{i}^+ \cup \delta_{i}^-) \cap choice_i =\emptyset$ with $\delta_i^+ \cup \delta_i^- \subseteq S_i^+ \cup tabu_i$, we write (7.2) using the premise of (7.1) as
	\begin{equation}
	tabu_{i+1}=(tabu_i \cup \delta_i^-) \setminus (S_i^+ \cup \delta_i^+) =tabu_i \setminus \delta_i^+=(A\setminus (S_i \cup S_{i}^+ \cup choice_i)) \setminus \delta_i^+.	
	\end{equation}
	Using (7.4) we rewrite (7.7) as
      \begin{equation}
	tabu_{i+1}=(A\setminus (S_i \cup S_{i}^+ \cup choice_{i+1} \cup \delta_i)) \setminus \delta_i^+=A \setminus (S_{i} \cup \delta_i \cup S_{i}^+ \cup \delta_i^+\cup choice_{i+1}).
	\end{equation}
	Using (2), (7.8) is equivalent to
	\begin{equation}
	tabu_{i+1} =A \setminus (S_{i+1} \cup S_{i+1}^+ \cup choice_{i+1}).
	\end{equation}
For the case where $\delta_i^+ \cup \delta_i^- \subseteq S_i^+ \cup tabu_i \cup choice_i$, let $\delta_i^+ = \delta_{it}^+ \cup \delta_{ic}^+ \cup \delta_{is}^+$ and $\delta_i^- = \delta_{it}^- \cup \delta_{ic}^- \cup \delta_{is}^-$ such that  $\delta_{it}^+ \cup \delta_{it}^- \subseteq tabu_i$, $\delta_{ic}^+ \cup \delta_{ic}^- \subseteq choice_i$, and $\delta_{is}^+ \cup \delta_{is}^- \subseteq S_i^+$. Then, we rewrite (7.2) as
\begin{equation}
\begin{array}{ll}
tabu_{i+1}&=(tabu_i \cup \delta_{ic}^- \cup \delta_{is}^- \cup \delta_{it}^-) \setminus (S_i^+ \cup \delta_{ic}^+ \cup \delta_{it}^+ \cup \delta_{is}^+)\\
&=(tabu_i \cup \delta_{ic}^- \cup \delta_{is}^-) \setminus (S_i^+ \cup \delta_{ic}^+ \cup \delta_{it}^+)\\
&=((A\setminus (S_i \cup S_i^+\cup choice_i)) \cup \delta_{ic}^- \cup \delta_{is}^-) \setminus (S_i^+ \cup \delta_{ic}^+ \cup \delta_{it}^+)\\
&=((A\setminus (S_i \cup S_i^+\cup choice_i)) \cup \delta_{ic}^-) \setminus (S_i^+ \cup \delta_{ic}^+ \cup \delta_{it}^+)\\
&=((A\setminus (S_i \cup S_i^+\cup choice_i)) \cup \delta_{ic}^-) \setminus (\delta_{ic}^+ \cup \delta_{it}^+)\\
&=((A\setminus (S_i \cup S_i^+\cup choice_{i+1} \cup \delta_i \cup \delta_{ic}^+ \cup \delta_{ic}^-)) \cup \delta_{ic}^-) \setminus (\delta_{ic}^+ \cup \delta_{it}^+)\\
&=(A\setminus (S_i \cup S_i^+\cup choice_{i+1} \cup \delta_i \cup \delta_{ic}^+)) \setminus (\delta_{ic}^+ \cup \delta_{it}^+)\\
&=(A\setminus (S_i \cup S_i^+\cup choice_{i+1} \cup \delta_i \cup \delta_{ic}^+ )) \setminus \delta_{it}^+\\
&=A\setminus (S_i \cup S_i^+\cup choice_{i+1} \cup \delta_i \cup \delta_{ic}^+ \cup \delta_{it}^+)\\
&=A\setminus (S_i \cup S_i^+\cup choice_{i+1} \cup \delta_i \cup \delta_{i}^+)\\
&=A\setminus (S_{i+1} \cup S_{i+1}^+\cup choice_{i+1}).
\end{array}
\end{equation}

	For the special case of line 17 in the algorithm, which indicates that $tabu_{i+1} =tabu_i \cup \{x\}_i$, $S_{i+1}=S_i$, and $choice_{i+1}=choice_i \setminus \{x\}_i$, we write
	\begin{equation}
	\begin{array}{ll}
	tabu_{i+1}&= tabu_i \cup \{x\}_i\\
	&=(A\setminus (S_i \cup S_i^+ \cup choice_i)) \cup \{x\}_i\\
	&= (A\setminus (S_{i+1} \cup S_{i+1}^+ \cup choice_{i})) \cup \{x\}_i\\
	&=(A \setminus (S_{i+1} \cup S_{i+1}^+ \cup choice_{i+1} \cup \{x\}_i)) \cup \{x\}_i\\
	&=A \setminus (S_{i+1} \cup S_{i+1}^+ \cup choice_{i+1}).
	\end{array}
	\end{equation}
	This concludes our proof of Proposition 7.\placeqed
\end{proof}

\begin{proposition}
	Let $H=(A,R)$ be an {\sc af} and Algorithm 1 be started with \begin{center}
		stb$(\emptyset, A \setminus \{x\mid (x,x)\in R\}, \{x\mid (x,x)\in R\})$.\end{center}
	The algorithm computes exactly the stable extensions of $H$.
\end{proposition}
\begin{proof}\renewcommand\qedsymbol{}
We will show that the following two statements (\emph{P1} \& \emph{P2}) hold.
	\begin{enumerate}	
		\item[\emph{P1.}] For every set, $S_i$, reported by Algorithm 1 at line 14 at some state~$i$, $S_i$ is a stable extension in $H$.
		\item[\emph{P2.}] For all $Q$, if $Q$ is a stable extension in $H$ and Algorithm 1 is sound (i.e. \emph{P1} is established), then there is a stable extension $S_i$, reported by the algorithm at line 14 at some state $i$, such that $S_i=Q$.
	\end{enumerate}
	Regarding \emph{P1}, to show that $S_i^+=A \setminus S_i$ we need to prove that
	\begin{equation}
	S_i \cap S_i^+ = \emptyset,
	\end{equation}
	and
	\begin{equation}
	S_i \cup S_i^+=A.
	\end{equation}	
	Note, (8.1) is proved in Proposition 5. For (8.2), it can be easily established by using Proposition 1 if we prove that for every state $i$
	\begin{equation}
	choice_i  \subseteq A \setminus (S_i \cup S_i^+\cup S_i^-),
	\end{equation}
	and
	\begin{equation}
	tabu_i=A \setminus (S_i\cup S_i^+ \cup choice_i).
	\end{equation}
	However, (8.3) and (8.4) are proved in Proposition 6 and 7 respectively. Thus, referring to line 14 in the algorithm, we note that $S_i$ is reported stable if and only if $tabu_i=choice_i=\emptyset$. Considering  (8.4), we note that $\emptyset = A \setminus (S_i \cup S_i^+ \cup \emptyset)$, and hence (8.2) holds. The proof of \emph{P1} is complete.\\ \\
Regarding \emph{P2}, we rewrite \emph{P2} (by modifying the consequence) into \emph{\'P2}.\\ \\
\emph{\'P2}: For all $Q$, if $Q$ is a stable extension in $H$ and Algorithm 1 is sound, then there is a stable extension $S_i$, reported by the algorithm at line 14 at some state $i$, such that for all $a \in Q$ it holds that $a \in S_i$. \\ \\
We establish \emph{\'P2} by contradiction. Later, we show that the consequence of \emph{\'P2} is equivalent to the consequence of \emph{P2}. Now, assume that \emph{\'P2} is false. \\ \\\
\emph{Negation of \'P2}: There is $Q$ such that $Q$ is a stable extension in $H$, Algorithm~1 is sound, and for every $S_i$ reported by the algorithm at line 14 at some state~$i$, there is $a \in Q$ such that $a \notin S_i$.\\ \\
We identify four cases.
\paragraph{Case 1} For $a \in choice_1$, if the algorithm terminates at line 2 during the very first execution of the repeat-until block (but not necessarily from the first round), then, since the algorithm is sound, $H$ has no stable extensions. This contradicts the assumption that $Q \supseteq \{a\}$ is a stable extension in~$H$. Hence, \emph{\'P2} holds.
\paragraph{Case 2} If $(a,a) \in R$, then this is a contradiction with the assumption that $Q\supseteq \{a\}$ is a stable extension in $H$. Hence, \emph{\'P2} holds. 
\paragraph{Case 3} With $a \in choice_1$, assume that after the very first execution of the repeat-until block (i.e. including one or more rounds), $a \in choice_k$ for some state $k\ge1$. Then, for a state $i\ge k$, let $x=a$ (see line 16 in Algorithm 1). If for all subsequent states $j > i$, the set $S_j \supseteq \{a\}$ is not reported stable by the algorithm, then, since the algorithm is sound, $a$ does not belong to any stable extension. This contradicts the assumption that $Q \supseteq \{a\}$ is a stable extension in $H$. Hence, \emph{\'P2} holds. 
\paragraph{Case 4} With $a \in choice_1$, assume that after the very first execution of the repeat-until block (i.e. including one or more rounds), $a \notin choice_i$ for some state $i>1$. This implies, according to the repeat-until block's actions, that $\{a\} \subseteq S_i \cup S_i^+ \cup tabu_i$. For $\{a\} \subseteq S_i$, if for all subsequent states $j > i$, the set $S_j \supseteq \{a\}$ is not reported stable by the algorithm, then, since the algorithm is sound, $a$ does not belong to any stable extension. This contradicts the assumption that $Q \supseteq \{a\}$ is a stable extension in $H$. Likewise, for $\{a\} \subseteq S_i^+ \cup tabu_i$, since the algorithm's actions are sound, this implies that $a$ does not belong to any stable extension. This contradicts the assumption that $Q \supseteq \{a\}$ is a stable extension in $H$. Hence, \emph{\'P2} holds.\\ \\
Now we rewrite the consequence of \emph{\'P2}.\\ \\
\emph{The consequence of \'P2}: There is a stable extension $S_i$, reported by the algorithm at line 14 at some state $i$, such that $Q \subseteq S_i$.\\ \\
$Q$ being a proper subset of $S_i$ is impossible because otherwise it contradicts that the algorithm is sound or that $Q$ is stable. Therefore, the consequence of \emph{\'P2} can be rewritten as next.\\ \\
\emph{The consequence of \'P2}: There is a stable extension $S_i$, reported by the algorithm at line 14 at some state $i$, such that $Q = S_i$.\\ \\
This is exactly the consequence of \emph{P2}. \placeqed
\end{proof}
\section{Validation of a finer-level implementation of the algorithm}
Recall that Algorithm 1 includes \emph{set} operations substantially. So far, it is left unspecified how to mechanize these operations. In this section, we will develop a new version that implements Algorithm 1. In other words, we will define at a finer level the underlying actions of Algorithm~1. To this end, we employ a total mapping to indicate the status of an argument in a given {\sc af} with respect to $S$, $S^+$, $choice$, and $tabu$.
 
 \begin{mydef}
 Let $H=(A,R)$ be an {\sc af}, $S\subseteq A$, $S \cap S^+=\emptyset$, $choice \subseteq A\setminus (S \cup S^+\cup S^-)$, $tabu = A\setminus (S\cup S^+\cup choice)$, and $\mu :A\rightarrow \{blank,in,out,must\mhyphen out\}$ be a total mapping. Then, $\mu$ is a labelling of $H$ with respect to $S$ if and only if for all $x \in A$ it is the case that
 \begin{equation} \nonumber
 \begin{array}{l}
 	x \in S \iff \mu(x)=in,\\
 	x \in S^+ \iff \mu(x)=out,\\
 	x \in choice \iff \mu(x)=blank, \text{ and }\\
 	x \in tabu \iff \mu(x)=must\mhyphen out.
\end{array}
 \end{equation}
\end{mydef}
Note that our labelling mechanism is used here as an \textit{ad hoc} algorithmic vehicle. For the widely-known labelling-based argumentation semantics, we refer the reader to \cite{DBLP:journals/sLogica/CaminadaG09,DBLP:journals/ker/BaroniCG11} for example. Now, we specify the labellings that correspond to stable extensions.
\begin{proposition}
	Let $H=(A,R)$ be an {\sc af}, $S\subseteq A$, $S\cap S^+=\emptyset$, $choice \subseteq A\setminus (S \cup S^+\cup S^-)$, $tabu = A\setminus (S\cup S^+\cup choice)$, and $\mu$ be a labelling of $H$ with respect to $S$. Then, $\{x\mid \mu(x)=in\}$ is a stable extension in $H$ if and only if $\{x\mid \mu(x)=blank\}=\emptyset$ and $\{x\mid \mu(x)=must\mhyphen out\}=\emptyset$. 
\end{proposition}
\begin{proof}\renewcommand\qedsymbol{}
We prove $\implies$ firstly. Considering Definition 1, $S=\{x\mid \mu(x)=in\}$ and $S^+=\{x\mid \mu(x)=out\}$. Suppose $\{x\mid \mu(x)=in\}$ is stable. Then,
\begin{flalign}
\{x\mid \mu(x)=out\}=& A \setminus \{x\mid \mu(x)=in\}.
\end{flalign}
Thus, 
 \begin{equation}
 \{x \mid \mu(x)=in\} \cup \{x \mid \mu(x)=out\}=A.
 \end{equation}
By Definition 1, observe that
 \begin{equation}
 \begin{array}{ll}
 tabu&=\{x\mid \mu(x)=must\mhyphen out\}\\
 &=A \setminus(\{x\mid \mu(x)=in\} \cup \{x\mid \mu(x)=out\} \cup \{x\mid \mu(x)=blank\}).
\end{array}
\end{equation}
Given (9.2), rewrite (9.3) as
 \begin{equation}  
 \{x\mid \mu(x)=must\mhyphen out\}=A \setminus(A \cup \{x\mid \mu(x)=blank\})=\emptyset.
 \end{equation}
 Likewise, as Definition~1 states that $choice=\{x\mid \mu(x)=blank\}$,
 \begin{equation}
  \{x\mid \mu(x)=blank\} \subseteq A \setminus (\{x\mid \mu(x)=in\} \cup \{x\mid \mu(x)=out\}\cup S^-).
 \end{equation}
Given (9.2), we rewrite (9.5) as
\begin{equation}
\{x\mid \mu(x)=blank\} \subseteq A \setminus (A\cup S^-).
\end{equation}
Hence, $\{x\mid \mu(x)=blank\}=\emptyset$.

Now we prove $\impliedby$. Suppose $\{x\mid \mu(x)=blank\}=\emptyset$ and $\{x\mid \mu(x)=must\mhyphen out\}=\emptyset$. Note that
\begin{equation}
\begin{array}{ll}
tabu&=\{x\mid \mu(x)=must\mhyphen out\}=\\
&=A \setminus (\{x \mid \mu(x)=in\} \cup \{x \mid \mu(x)=out\} \cup \{x \mid \mu(x)=blank\}),
\end{array}
\end{equation}
which can be rewritten as
\begin{equation}
\emptyset= A \setminus (\{x \mid \mu(x)=in\} \cup\{x \mid \mu(x)=out\} \cup \emptyset).
\end{equation}
And subsequently,
\begin{equation}
A = \{x \mid \mu(x)=in\} \cup \{x \mid \mu(x)=out\}.
\end{equation}
Given (9.9), and since $\{x \mid \mu(x)=in\} \cap \{x \mid \mu(x)=out\} = \emptyset$, we note that $\{x\mid \mu(x)=out\}=A\setminus \{x\mid \mu(x)=in\}$. Consequently, $\{x \mid \mu(x)=in\}$ is a stable extension in $H$.\placeqed
\end{proof}

Using the labelling notion introduced in Definition 1, next we discuss three propositions that are analogous to Proposition 2, 3 and 4, and which are crucial for efficient enumeration of stable extensions.

\begin{proposition}
	 Let $H=(A,R)$ be an {\sc af}, $S\subseteq A$, $S \cap S^+=\emptyset$, $choice \subseteq A\setminus (S \cup S^+\cup S^-)$, $tabu = A\setminus (S\cup S^+\cup choice)$, $\mu$ be a labelling of $H$ with respect to $S$, and $v$ be an argument with $\mu(v)=blank$ such that for all $y \in \{v\}^-$ it holds that $\mu(y) \in \{out,must\mhyphen out\}$. If there is a stable extension $T \supseteq \{x\mid \mu(x)=in\}$ such that $T \setminus \{x\mid \mu(x)=in\} \subseteq \{x\mid \mu(x)=blank\}$, then $v \in T$.
\end{proposition}
\begin{proof}\renewcommand\qedsymbol{}
	Suppose $v \notin T$. Then, $v \in T^+$ because $T$ is stable. Thus, 
	\begin{equation}
	\exists y \in \{v\}^- \text{ such that } y \in T. 
	\end{equation}
	However, consistently with the premise of Proposition 10,
	\begin{equation}
	\forall y \in \{v\}^-~\mu(y) \in \{out,must\mhyphen out\}.
	\end{equation}
	Therefore, (10.1) and (10.2) imply
	\begin{equation}
	T \cap (\{x\mid \mu(x)=out\} \cup \{x\mid \mu(x)=must\mhyphen out\}) \neq \emptyset.
	\end{equation}
	According to the premise of this proposition, we note that
	\begin{equation}
	S \cap (S^+ \cup tabu)=\emptyset \text{ and }\\
	choice \cap (S^+ \cup tabu)=\emptyset.
	\end{equation}
	Applying Definition 1, (10.4) can be rewritten respectively as
	\begin{equation}	
	\begin{array}{l}
	\{x\mid \mu(x)=in\} \cap (\{x\mid \mu(x)=out\}\cup \{x\mid \mu(x)=must\mhyphen out\})=\emptyset, 	\text{ and }\\
	\{x\mid \mu(x)=blank\} \cap (\{x\mid \mu(x)=out\}\cup \{x\mid \mu(x)=must\mhyphen out\})=\emptyset.
	\end{array}
	\end{equation}
	Given the premise of this proposition, note that 
	\begin{equation}
	T\supseteq \{x\mid \mu(x)=in\}, \text {and }
	T\setminus \{x\mid \mu(x)=in\} \subseteq \{x\mid \mu(x)=blank\}.
	\end{equation}
Thus, (10.5) and (10.6) imply that
	\begin{equation}
	T \cap (\{x\mid \mu(x)=out\}\cup \{x\mid \mu(x)=must\mhyphen out\}) =\emptyset.
	\end{equation}
	Observe the contradiction between (10.3) and (10.7). \placeqed
\end{proof}
\begin{proposition}
	Let $H=(A,R)$ be an {\sc af}, $S\subseteq A$, $S \cap S^+=\emptyset$, $choice \subseteq A\setminus (S \cup S^+\cup S^-)$, $tabu = A\setminus (S\cup S^+\cup choice)$, $\mu$ be a labelling of $H$ with respect to $S$, and $v$ be an argument with $\mu(v)=blank$ such that for some $y$ with $\mu(y)=must\mhyphen out$ it is the case that $\{x\in \{y\}^- \mid \mu(x)=blank\}=\{v\}$. If there is a stable extension $T \supseteq \{x\mid \mu(x)=in\}$ such that $T \setminus \{x\mid \mu(x)=in\} \subseteq \{x\mid \mu(x)=blank\}$, then $v \in T$. 	
\end{proposition}
\begin{proof}\renewcommand\qedsymbol{}
	Suppose $v \notin T$. Then,
	\begin{equation}
	T \setminus \{x\mid \mu(x)=in\} \subseteq \{x\mid \mu(x)=blank\} \setminus \{v\}. 
	\end{equation}
	Due to $tabu=A\setminus(S\cup S^+ \cup choice)$, and according to Definition 1 it holds that
	\begin{equation}
	\begin{array}{ll}
	tabu&=\{x\mid \mu(x)=must\mhyphen out\} \\
	&= A\setminus (\{x\mid \mu(x)=in\}\cup \{x\mid \mu(x)=out\}\cup \{x\mid \mu(x)=blank\}).
	\end{array}
	\end{equation}
	Since $T \supseteq \{x\mid \mu(x)=in\}$ such that $T \setminus \{x\mid \mu(x)=in\} \subseteq \{x\mid \mu(x)=blank\}$,
	\begin{flalign}
	T \subseteq \{x\mid \mu(x)=in\} \cup \{x\mid \mu(x)=blank\}. 
	\end{flalign}
	Referring to the premise of this proposition, as $y \in \{x\mid \mu(x)=must\mhyphen out\}$, $y \in tabu$ (recall Definition~1) and so $y \notin T$. However, since $T$ is stable, $y \in T^+$. Thus, 
	\begin{equation}
	\{y\}^- \cap T \neq \emptyset.
	\end{equation}
	Given (11.2) and because $y \in \{x\mid \mu(x)=must\mhyphen out\}$, it holds that $y \notin \{x\mid \mu(x)=out\}$; hence $ y \notin S^+$ consistently with Definition 1. Therefore, \begin{equation}\{y\}^- \cap \{x\mid \mu(x)=in\} = \emptyset.\end{equation} Thereby, (11.1), (11.4) and (11.5) together imply that
	\begin{equation}
	\{y\}^- \cap (\{x\mid \mu(x)=blank\} \setminus \{v\}) \neq \emptyset.	
	\end{equation} 
	From the premise of this proposition, it is the case that $\{x\in \{y\}^- \mid \mu(x)=blank\}=\{v\}$. Thus,
	\begin{equation}
	\{y\}^- \cap \{x\mid \mu(x)=blank\} =\{v\}.	
	\end{equation}
	See the contradiction between (11.6) and (11.7).\placeqed
\end{proof}

\numberwithin{equation}{proposition}
\begin{proposition}
	Let $H=(A,R)$ be an {\sc af}, $S\subseteq A$, $S \cap S^+=\emptyset$, $choice \subseteq A\setminus (S \cup S^+\cup S^-)$, $tabu = A\setminus (S\cup S^+\cup choice)$, $\mu$ be a labelling of $H$ with respect to $S$, and $x$ be an argument with $\mu(x)=must\mhyphen out$ such that for all $y \in \{x\}^-$ it is the case that $\mu(y)\in \{out,must\_out\}$. There does not exist a stable extension $Q \supseteq \{v\mid \mu(v)=in\}$ such that $Q \setminus \{v\mid \mu(v)=in\} \subseteq \{v\mid \mu(v)=blank\}$.
\end{proposition}
\begin{proof}\renewcommand\qedsymbol{}
Consider Definition~1 during the proof. As $x\in \{v\mid \mu(v)=must\mhyphen out\}$ and $tabu = A\setminus (S\cup S^+\cup choice)$,
 \begin{equation}
 x \notin \{v\mid \mu(v)=in\} \cup \{v\mid \mu(v)=out\} \cup \{v \mid \mu(v)=blank\}.
 \end{equation}
Subsequently, 
 \begin{equation}\begin{array}{l}
 \forall T \supseteq \{v\mid \mu(v)=in\} \text{ such that } T \setminus \{v\mid \mu(v)=in\} \subseteq \{v\mid \mu(v)=blank\},\\
 x\notin T.
\end{array}\end{equation}
 Since for all $y \in \{x\}^-$ it holds that $\mu(y) \in \{out,must\_out\}$, 
 \begin{equation}
 \{x\}^- \subseteq \{v\mid \mu(v)=out\} \cup \{v\mid \mu(v)=must\mhyphen out\}.
 \end{equation}
Hence,
\begin{equation}
\{x\}^- \cap (\{v\mid \mu(v)=in\} \cup \{v\mid \mu(v)=blank\}) = \emptyset.
\end{equation} 
Therefore, 
\begin{equation}\begin{array}{l}
\forall T \supseteq \{v\mid \mu(v)=in\} \text{ such that } T \setminus \{v\mid \mu(v)=in\} \subseteq \{v\mid \mu(v)=blank\},\\
 \{x\}^- \cap T= \emptyset.
\end{array}\end{equation} 
Thus,
 \begin{equation}\begin{array}{l}
  \forall T  \supseteq \{v\mid \mu(v)=in\} \text{ such that } T \setminus \{v\mid \mu(v)=in\} \subseteq \{v\mid \mu(v)=blank\},\\
  x\notin T^+.
 \end{array}\end{equation}
 Given (12.2) and (12.6), 
 \begin{equation}\begin{array}{l}
 \forall T  \supseteq \{v\mid \mu(v)=in\} \text{ such that } T \setminus \{v\mid \mu(v)=in\} \subseteq \{v\mid \mu(v)=blank\},\\
 T^+~\not = A \setminus T.
 \end{array}\end{equation} 
 Consequently,
 \begin{equation}\begin{array}{l}
 \forall T  \supseteq \{v\mid \mu(v)=in\} \text{ such that } T \setminus \{v\mid \mu(v)=in\} \subseteq \{v\mid \mu(v)=blank\},\\
 T \text{ is not stable.}
 \end{array}\end{equation} 
That completes the proof of Proposition 12.\placeqed
\end{proof}

Now, we turn to the lines 2, 3 and~7 in Algorithm 1. Recall, by applying these lines it is required to search (respectively) for an argument $x$ such that

\begin{equation}
\begin{array}{l}
\nonumber
x \in tabu \text{ with } \{x\}^- \subseteq S^+ \cup tabu, \text{ or } \\
x \in choice \text{ with } \{x\}^- \subseteq S^+ \cup tabu, \text{ or } \\
x \in tabu \text{ with } |\{x\}^- \cap choice|=1.
\end{array}
\end{equation}
To implement these lines (i.e. 2, 3, and~7) efficiently, we define the following construct (inspired by a hint given in \cite{DBLP:books/sp/09/ModgilC09} and already utilized in \cite{DBLP:journals/cj/NofalAD21}). 
\begin{mydef}
	Let $H=(A,R)$ be an {\sc af}, $S\subseteq A$, $S\cap S^+=\emptyset$, $choice \subseteq A\setminus (S \cup S^+\cup S^-)$, $tabu = A\setminus (S\cup~S^+\cup~choice)$, and $\mu$ be a labelling of $H$ with respect to $S$. For all $x \in A$, if $\mu(x) \in \{blank,must\mhyphen out\}$, then $\pi(x)\myeq |\{y \in \{x\}^-: \mu(y)=blank\}|$.
\end{mydef}
The intuition behind Definition 2 is that instead of checking (recurrently over and over) the whole set of attackers, $\{x\}^-$, of a given argument $x$, to see whether $\{x\}^-$ is contained in $S^+ \cup tabu$, one might hold a counter of the attackers of $x$ that are currently in $choice$. Whenever an attacker of $x$ is being moved from $choice$, we decrease the counter. So, for checking whether $\{x\}^- \subseteq S^+ \cup tabu$ all we need to do is to test if the counter is equal to zero. In the following proposition, we prove the usage of this notion (i.e. counting attackers).
\numberwithin{equation}{proposition}
\begin{proposition}
	Let $H=(A,R)$ be an {\sc af}, $S\subseteq A$, $S\cap S^+=\emptyset$, $choice \subseteq A\setminus (S \cup S^+\cup S^-)$, $tabu = A\setminus (S\cup~S^+\cup~choice)$, and $\mu$ be a labelling of $H$ with respect to $S$. Then, for all $x\in A$ it holds that
	\begin{flalign} 
	x \in tabu \wedge \{x\}^- \subseteq S^+ \cup tabu \iff \mu(x)=must\mhyphen out \wedge \pi(x)=0,\\
	x \in choice \wedge \{x\}^- \subseteq S^+ \cup tabu \iff \mu(x)=blank \wedge \pi(x)=0, \\
	x \in tabu \wedge |\{x\}^- \cap choice|=1 \iff	\mu(x)=must\mhyphen out \wedge \pi(x)=1.
	\end{flalign}
\end{proposition}
\begin{proof}\renewcommand\qedsymbol{}
	Consider Definition 1 and 2 throughout this proof. From the premise of this proposition, note that $S \cup S^+ \cup choice \cup tabu = A$. Additionally, the sets: $S$, $S^+$, $tabu$, and $choice$ are pairwise disjoint. Thus, for all $x$ 
	\begin{equation}
	\begin{array}{l}
	\nonumber
	x \in tabu \wedge \{x\}^- \subseteq S^+ \cup tabu \iff \\
	x \in tabu  \wedge \{x\}^- \cap (S \cup choice)=\emptyset \iff \\
	x \in tabu \wedge x \notin S^+ \wedge \{y \in \{x\}^- \mid y \in choice\}=\emptyset \iff \\
	\mu(x)=must\mhyphen out \wedge \{y \in \{x\}^- \mid \mu(y)=blank\}=\emptyset \iff \\
	\mu(x)=must\mhyphen out \wedge \pi(x)=0.
	\end{array}
	\end{equation}
	That completes the proof of (13.1). Likewise, (13.2) is true since for all $x$ it holds that
	\begin{equation}
	\begin{array}{l}
	\nonumber
	x \in choice \wedge \{x\}^- \subseteq S^+ \cup tabu \iff \\
	x \in choice  \wedge \{x\}^- \cap (S \cup choice)=\emptyset \iff \\
	x \in choice \wedge x \notin S^+ \wedge \{y \in \{x\}^- \mid y \in choice\}=\emptyset \iff \\
	\mu(x)=blank \wedge \{y \in \{x\}^- \mid \mu(y)=blank\}=\emptyset \iff \\
	\mu(x)=blank \wedge \pi(x)=0.
	\end{array}
	\end{equation}
	As to (13.3), for all $x$ it is the case that
	\begin{equation}
	\begin{array}{l}
	\nonumber x\in tabu \wedge |\{x\}^- \cap choice|=1 \iff \\
	\mu(x)=must\mhyphen out \wedge|\{y \in \{x\}^-: y \in choice\}|=1 \iff \\
	\mu(x)=must\mhyphen out \wedge |\{y \in \{x\}^-: \mu(y)=blank\}|=1 \iff \\ \mu(x)=must\mhyphen out \wedge \pi(x)=1.\text{\placeqed}
	\end{array}
	\end{equation}	 	
\end{proof}
Making use of the structures $\mu$ and $\pi$ specified in Definition 1 \& 2 respectively, we give Algorithm 2.  Let $H=(A,R)$ be an {\sc af}, $\mu: A \rightarrow \{in,out,must\mhyphen out,blank\}$ be a total mapping such that for all $x \in A$,
	\[
	\mu(x)=
	\begin{cases}
	must\mhyphen out, &~if~(x,x)\in R;\\
	blank, &~otherwise.\\
	\end{cases}
	\]
And let $\pi: A \rightarrow \{0,1,2,...,|A|\}$ be a total mapping such that for all $x \in A$, $\pi(x)=|\{y \in \{x\}^-: \mu(y)=blank\}|$. By invoking Algorithm 2 with list-stb-ext($\mu$, $\pi$, $\{x \mid \mu(x)=blank \wedge \pi(x)=0\}$), the algorithm computes all stable extensions in $H$.
\begin{algorithm}
	\small{
	\caption{list-stb-ext($\mu$, $\pi$, $\gamma$)}
	\While{$\gamma \not = \emptyset$}{		
		For some $q \in \gamma$ \textbf{do} $\gamma \gets \gamma \setminus \{q\}$;~$\mu(q) \gets in$\;
		\textbf{foreach} $z \in \{q\}^+$ with $\mu(z)=must\mhyphen out$ \textbf{do} $\mu(z)\gets out$\;
		\ForEach{$z \in \{q\}^- \cup \{q\}^+$ with $\mu(z)=blank$}{
			\textbf{if} $z \in \{q\}^+$ \textbf{then} $\mu(z)\gets out$ \textbf{else} $\mu(z)\gets must\mhyphen out$\;
			\ForEach{$x \in \{z\}^+$}{
				$\pi(x) \gets \pi(x)-1$\;
				\textbf{if} $\mu(x)=must\mhyphen out$ with $\pi(x)=0$ \textbf{then} return\;
				\textbf{if} $\mu(x)=blank$ with $\pi(x)=0$ \textbf{then} $\gamma \gets \gamma \cup \{x\}$\;
				\If{$\mu(x)=must\mhyphen out~with~\pi(x)=1$}{
					$\gamma \gets \gamma \cup \{y \in \{x\}^- \mid \mu(y)=blank\}$\;
				}
			}
		}
	}	
	\If{$\{x \mid \mu(x)=blank\}=\emptyset$}{
		\textbf{if} $\{x \mid \mu(x)=must\mhyphen out\}=\emptyset$ \textbf{then} $\{x \mid \mu(x)=in\}$ is stable\;	
		return\;		
	}
	
	list-stb-ext($\mu$, $\pi$, $\{x\}$);~~\tcp{for some $x$ with $\mu(x)=blank$}
	$\mu(x) \gets must\mhyphen out$\; 
	\ForEach{$z \in \{x\}^+$}{
		$\pi(z) \gets \pi(z)-1$\;
		\textbf{if} $\mu(z)=must\mhyphen out$ with $\pi(z)=0$ \textbf{then} return\;
		\textbf{if} $\mu(z)=blank$ with $\pi(z)=0$ \textbf{then} $\gamma \gets \gamma \cup \{z\}$\;
		\textbf{if} $\mu(z)=must\mhyphen out$ and $\pi(z)=1$ \textbf{\footnotesize{then}} $\gamma \gets \gamma \cup \{y\in \{z\}^-\mid \mu(y)=blank\}$\;		
	}
	list-stb-ext($\mu$, $\pi$, $\gamma$)\;
}
\end{algorithm}

\noindent \textbf{Example 3.} Using Algorithm 2, we list the stable extensions of $H_1$ (see Figure~1). So, we start the algorithm with list-stb-ext($\mu$, $\pi$, $\gamma$) such that
\begin{equation}
\begin{array}{l}
	\mu = \{(a,blank),(b,blank),(c,blank),(d,blank),(e,blank),(f,blank)\},\\
	\pi = \{(a,2),(b,2),(c,2),(d,1),(e,1),(f,2)\},\\
	\gamma=\emptyset.
\end{array} \tag{3.1}
\end{equation}
Referring to line 15 in Algorithm 2, let $x$ be the argument $a$. Then, invoke list-stb-ext($\mu$, $\pi$, $\gamma$) with 
 \begin{equation}
 \begin{array}{l}
 	\mu = \{(a,blank),(b,blank),(c,blank),(d,blank),(e,blank),(f,blank)\},\\
 	\pi = \{(a,2),(b,2),(c,2),(d,1),(e,1),(f,2)\},\\
 	\gamma= \{a\}.
 \end{array} \tag{3.2}
 \end{equation}  
Apply a first round of the \textit{while} loop of the algorithm. Therefore,
 \begin{equation}
\begin{array}{l}
\mu = \{(a,in),(b,out),(c,blank),(d,blank),(e,must\mhyphen out),(f,must\mhyphen out)\},\\
\pi = \{(a,0),(b,2),(c,0),(d,0),(e,1),(f,1)\},\\
\gamma= \{c,d\}.
\end{array} \tag{3.3}
\end{equation}
After a second and third round of the \textit{while} loop, 
\begin{equation}
\begin{array}{l}
\mu = \{(a,in),(b,out),(c,in),(d,in),(e,out),(f,out)\},\\
\pi = \{(a,0),(b,2),(c,0),(d,0),(e,1),(f,1)\},\\
\gamma= \emptyset.
\end{array}\tag{3.4}
\end{equation}
Now, $\{a,c,d\}$ is stable, see line 13 in Algorithm 2. Applying line 14, backtrack to state (3.1). Perform the actions at lines 16--21, and then invoke list-stb-ext($\mu$, $\pi$, $\gamma$) (line~22) such that
\begin{equation}
\begin{array}{l}
\mu = \{(a,must\mhyphen out),(b,blank),(c,blank),(d,blank),(e,blank),(f,blank)\},\\
\pi = \{(a,2),(b,1),(c,2),(d,1),(e,1),(f,2)\},\\
\gamma= \emptyset.
\end{array}\tag{3.5}
\end{equation}
Referring to line 15, let $x$ be the argument $b$. Invoke list-stb-ext($\mu$, $\pi$, $\gamma$) with
\begin{equation}
\begin{array}{l}
\mu = \{(a,must\mhyphen out),(b,blank),(c,blank),(d,blank),(e,blank),(f,blank)\},\\
\pi = \{(a,2),(b,1),(c,2),(d,1),(e,1),(f,2)\},\\
\gamma= \{b\}.
\end{array}\tag{3.6}
\end{equation}
Apply a first round of the \textit{while} loop. Thereby,
\begin{equation}
\begin{array}{l}
\mu = \{(a,must\mhyphen out),(b,in),(c,out),(d,out),(e,blank),(f,blank)\},\\
\pi = \{(a,2),(b,0),(c,2),(d,1),(e,0),(f,1)\},\\
\gamma= \{e\}.
\end{array}\tag{3.7}
\end{equation}
As $\gamma \not= \emptyset$, apply a second round of the \textit{while} loop. Thus,
\begin{equation}
\begin{array}{l}
\mu = \{(a,out),(b,in),(c,out),(d,out),(e,in),(f,out)\},\\
\pi = \{(a,1),(b,0),(c,2),(d,1),(e,0),(f,1)\},\\
\gamma= \emptyset.
\end{array}\tag{3.8}
\end{equation}
Referring to line 13 in Algorithm 2, $\{b,e\}$ is stable. Applying line 14, backtrack to state (3.5). Afterwards, apply the lines 16--21 and then invoke list-stb-ext($\mu$, $\pi$, $\gamma$) (line 22) with
\begin{equation}
\begin{array}{l}
\mu = \{(a,must\mhyphen out),(b,must\mhyphen out),(c,blank),(d,blank),(e,blank),(f,blank)\},\\
\pi = \{(a,2),(b,1),(c,1),(d,0),(e,1),(f,2)\},\\
\gamma= \{d\}.
\end{array}\tag{3.9}
\end{equation}
In performing a first round of the \textit{while} loop, we get
\begin{equation}
\begin{array}{l}
\mu = \{(a,must\mhyphen out),(b,out),(c,blank),(d,in),(e,out),(f,out)\},\\
\pi = \{(a,0),(b,1),(c,0),(d,0),(e,1),(f,1)\},\\
\gamma= \{c,f\}.
\end{array}\tag{3.10}
\end{equation}
However, as $\mu(a)=must\mhyphen out$ and $\pi(a)=0$, by applying line 8 we return to a previous state and eventually terminate the procedure.\\

Next, we will give four more propositions that (along with Proposition 9, 10, 11, 12, and 13) will establish the correctness of Algorithm 2. To this end, we denote by $T_i$ the elements of a set $T$ at the algorithm's state $i$. Algorithm~2 enters a new state whenever line 2 or line 16 are executed. Note that lines 2 and 16 are sensible to be selected to designate a beginning of a new state of the algorithm since they include a decision to re-map an argument to $in$ or $must\_out$ rather than an imposed re-mapping under some conditions, such as those conditions in the lines 3--5. Focusing on the arguments that are mapped to $in$, in the initial state of the algorithm, we let
\setcounter{equation}{2}
\counterwithout{equation}{proposition}
\begin{equation}
\{x|\mu_{1}(x)=in\}=\emptyset,
\end{equation} 
and for all states $i$ it holds that
\begin{equation}
\begin{array}{l}
\{x|\mu_{i+1}(x)=in\}= \{x|\mu_{i}(x)=in\}  \text{ (see line 16) or } \\
\{x|\mu_{i+1}(x)=in\}= \{x|\mu_{i}(x)=in\} \cup \{q\}_i
\end{array}
\end{equation}
such that $\{q\}_i$ is a one-element set containing an argument from
\begin{equation}
\{x \mid \mu_i(x)=blank\} \text{ (see line 15 in Algorithm 2)}
\end{equation}
or an argument from
\begin{equation}
\{x \text{ with } \mu_{i}(x)=blank \mid \pi_{i}(x)=0\} \text{ (see lines 9 and 20 in Algorithm 2)}
\end{equation}
or an argument from
\begin{equation}
 \{x \text{ with } \mu_{i}(x)=blank \mid \exists y \in \{x\}^+: \mu_{i}(y)=must\mhyphen out  \wedge \pi_{i}(y)=1\},
\end{equation}
see lines 10, 11 and 21 in Algorithm 2.
 \numberwithin{equation}{proposition}
\begin{proposition}
	Let $H=(A,R)$ be an {\sc af} and $\pi: A \rightarrow \{0,1,2,...,|A|\}$ be a total mapping such that for all $x \in A$, $\pi_1(x)=|\{x\}^- \setminus \{y: (y,y) \in R\}|$. And let $\mu: A \rightarrow \{in,out,must\mhyphen out,blank\}$ be a total mapping such that for all $x \in A$,
	\[
	\mu_1(x)=
	\begin{cases}
	must\mhyphen out, &~if~(x,x)\in R;\\
	blank, &~otherwise.\\
	\end{cases}
	\]
	And assume that Algorithm 2 is started with list-stb-ext($\mu_1$, $\pi_1$, $\{x \mid \mu_1(x)=blank \wedge \pi_1(x)=0\}$). For every state $i$ it holds that $\{x \mid \mu_i(x)=in\} \cap \{x \mid \mu_i(x)=in\}^+ = \emptyset$.
\end{proposition}
\begin{proof}\renewcommand\qedsymbol{}
	Since $\{x\mid \mu_{1}(x)=in\}=\emptyset$, $\{x\mid \mu_{1}(x)=in\} \cap \{x\mid \mu_{1}(x)=in\}^+=\emptyset$. We now show that for every state $i$
	\begin{equation}
	\begin{array}{l}
	\{x\mid \mu_i(x)=in\} \cap \{x \mid \mu_{i}(x)=in\}^+=\emptyset \implies\\ \{x\mid \mu_{i+1}(x)=in\} \cap \{x \mid \mu_{i+1}(x)=in\}^+=\emptyset.
	\end{array}
	\end{equation}
	Focusing on the consequence of (14.1) and using (4), we need to prove that
	\begin{equation}
	\begin{array}{l}
\{x\mid \mu_{i+1}(x)=in\} \cap \{x \mid \mu_{i+1}(x)=in\}^+= \\
(\{x\mid \mu_{i}(x)=in\} \cup \{q\}_i) \cap (\{x\mid \mu_{i}(x)=in\}^+ \cup \{q\}_{i}^+) =\emptyset.
	\end{array}
	\end{equation}	
	Subsequently, to establish (14.2), we need to show that for any state $i$
	\begin{flalign}	
	\{x\mid \mu_{i}(x)=in\} \cap \{x \mid \mu_{i}(x)=in\}^+= \emptyset,\\
	\{x\mid \mu_{i}(x)=in\} \cap \{q\}_i^+= \emptyset,\\
	\{q\}_i \cap \{x \mid \mu_{i}(x)=in\}^+=\emptyset, \\
	\nonumber \text{ and }\\
	\{q\}_i \cap \{q\}_i^+ = \emptyset.
	\end{flalign}
	
	Assuming the premise of (14.1), (14.3) follows immediately. \\
	As to (14.4), note that $\mu_i(q)=blank$ for all states $i$, see (4)--(7). Now, suppose that (14.4) is false. Thus, 
	\[\text{ at some state } i,~\exists x \text{ with } \mu_i(x)=in \text{ such that } (q,x) \in R \text{ and } \mu_i(q)=blank.\]
	This means 
	\[\text{ at some state } i,~\exists x \text{ with } \mu_i(x)=in \text{ such that } \exists y \in \{x\}^- \text{ with }\mu_i(y)=blank.\]
	This contradicts the actions of Algorithm 2 (lines 2--5) that indicate
	\[\text{ for all states }i,~\forall x \text{ with } \mu_i(x)=in, \forall y \in \{x\}^-, \mu_i(y) \in \{out,must\_out\}.\]
	Therefore, (14.4) holds. 
	
	Now we prove (14.5). Observe, for any state $i$, $\mu_{i}(q)=blank$, see (4)--(7). Assume that (14.5) is false. Thereby,
	\[\text{ at some state } i,~\exists x \text{ with } \mu_i(x)=in \text{ such that }(x,q)\in R \text{ and } \mu_i(q)=blank.\]
	Thus, 
	\[\text{ at some state }i,~\exists x \text{ with } \mu_i(x)=in \text{ such that } \exists y \in \{x\}^+ \text{ with } \mu_i(y)=blank.\]
	This contradicts the actions of Algorithm 2 (lines 2--5) that require
	\[\text{ for all states }i,~\forall x \text{ with } \mu_i(x)=in, \forall y \in \{x\}^+, \mu_i(y) =out.\]
	Therefore, (14.5) holds. \\
	Now we show (14.6). Recall that (see the conditions of this proposition)
	\[\forall x \text{ with } (x,x) \in R, \mu_1(x)=must\_out,\]
	and so the actions of Algorithm 2 collectively (especially line 3) entail that
	\[\forall x \text{ with } (x,x) \in R, \text{ for all states } i,~\mu_i(x) \in \{out,must\_out\}.\] 
	Now, suppose that (14.6) is false. Thus, $(q,q) \in R$ and $\mu_i(q)=blank$ at some state $i$. Contradiction. \placeqed	
\end{proof}

\begin{proposition}
	Let $H=(A,R)$ be an {\sc af} and $\pi: A \rightarrow \{0,1,2,...,|A|\}$ be a total mapping such that for all $x \in A$, $\pi_1(x)=|\{x\}^- \setminus \{y: (y,y) \in R\}|$. And let $\mu: A \rightarrow \{in,out,must\mhyphen out,blank\}$ be a total mapping such that for all $x \in A$,
	\begin{equation}
	\mu_1(x)=
	\begin{cases}
	must\mhyphen out, &~if~(x,x)\in R;\\
	blank, &~otherwise.\\
	\end{cases}
	\end{equation}	
	And assume that Algorithm 2 is started with list-stb-ext($\mu_1$, $\pi_1$, $\{x \mid \mu_1(x)=blank \wedge \pi_1(x)=0\}$). For every state $i$ it holds that 
	\begin{equation}
	\begin{array}{l}
	\{x \mid \mu_i(x)=blank\} \subseteq\\
	 A \setminus (\{x\mid \mu_i(x)=in\} \cup \{x \mid \mu_i(x)=in\}^+ \cup \{x \mid \mu_i(x)=in\}^-).
	\end{array}
	\end{equation}
\end{proposition}
\begin{proof}\renewcommand\qedsymbol{}
	Since $\mu$ is a total mapping from $A$ to $\{in,out,must\mhyphen out,blank\}$, for every state $i$ it holds that $\{x\mid \mu_i(x)=blank\} \subseteq A$. Nonetheless, to establish (15.2) we need to check that for every state $i$
	\begin{equation}
	\begin{array}{l}
	\{x\mid \mu_i(x)=blank\} \cap \\
	(\{x\mid \mu_i(x)=in\} \cup \{x\mid \mu_i(x)=in\}^+ \cup \{x\mid \mu_i(x)=in\}^-) = \emptyset.
	\end{array}
	\end{equation}
	According to the algorithm's actions (lines 2--5), note that for every state $i$
	\[\forall x \text{ with } \mu_i(x)=in,\]
	\[\forall y \in\{x\}^+~\mu_i(y)=out, \text{ and } \forall z \in\{x\}^-~\mu_i(z)\in\{out,must\_out\}.\]
	Thus, for every state $i$
	\begin{flalign}
	\forall y \in \{x\mid \mu_i(x)=in\}^+~\mu_i(y)=out, \\
	\forall z \in \{x\mid \mu_i(x)=in\}^-~\mu_i(z)\in \{out,must\mhyphen out\}. 	
	\end{flalign}
	Rewrite (15.4) and (15.5) respectively as
	\begin{flalign}
	\{x\mid \mu_i(x)=in\}^+ \subseteq \{x\mid \mu_i(x)=out\}, \\
	\{x\mid \mu_i(x)=in\}^- \subseteq \{x\mid \mu_i(x)\in \{out,must\mhyphen out\}\}.
	\end{flalign}
	Thus,
	\begin{flalign}
	\nonumber \{x\mid \mu_i(x)=in\}^+ \cup \{x\mid \mu_i(x)=in\}^-\subseteq \\ \{x\mid \mu_i(x)=out\} \cup \{x\mid \mu_i(x)=must\mhyphen out\}.
	\end{flalign}
	Since $\mu:A\to \{blank,in,out,must\mhyphen out\}$ is a total mapping, for all $i$
	\begin{equation}
	\begin{array}{l}
	\{x\mid \mu_i(x)=blank\} \cap \\
	(\{x\mid \mu_i(x)=in\} \cup \{x\mid \mu_i(x)=out\} \cup \{x\mid \mu_i(x)=must\mhyphen out\})=\emptyset.
	\end{array}
	\end{equation}
	Due to (15.8) and (15.9), (15.3) holds. \placeqed	
\end{proof}

\begin{proposition}
	Let $H=(A,R)$ be an {\sc af} and $\pi: A \rightarrow \{0,1,2,...,|A|\}$ be a total mapping such that for all $x \in A$, $\pi_1(x)=|\{x\}^- \setminus \{y: (y,y) \in R\}|$. And let $\mu: A \rightarrow \{in,out,must\mhyphen out,blank\}$ be a total mapping such that for all $x \in A$,
	\begin{equation}
	\mu_1(x)=
	\begin{cases}
	must\mhyphen out, &~if~(x,x)\in R;\\
	blank, &~otherwise.\\
	\end{cases}
	\end{equation}	
	And assume that Algorithm 2 is started with list-stb-ext($\mu_1$, $\pi_1$, $\{x \mid \mu_1(x)=blank \wedge \pi_1(x)=0\}$). For every state $i$ it holds that 
	\begin{equation}
	\begin{array}{l}
 \{x \mid \mu_i(x)=must\mhyphen out\} = \\
 A \setminus (\{x\mid \mu_i(x)=in\} \cup \{x \mid \mu_i(x)=in\}^+ \cup \{x \mid \mu_i(x)=blank\}).
	\end{array}
	\end{equation}
\end{proposition}
\begin{proof}\renewcommand\qedsymbol{}
Note that according to the actions of Algorithm 2 (lines 2--5), for every state $i$ it is the case that
\[\forall x \text{ with } \mu_i(x)=in, \forall y \in \{x\}^+~\mu_i(y)=out.\]
Subsequently, for every state $i$ it holds that
\begin{equation}
\forall y \in \{x\mid \mu_i(x)=in\}^+~\mu_i(y)=out.
\end{equation}
Using (16.3) to rephrase (16.2), we need to show that for all $i$
\begin{equation}
 \{x \mid \mu_i(x)=must\mhyphen out\} = A \setminus \{x\mid \mu_i(x) \in\{in,out,blank\}\}.
\end{equation}
Nonetheless, since $\mu:A\to \{in,out,blank,must\mhyphen out\}$ is a total mapping, we note that for all $i$ 
\begin{equation}\begin{array}{l}
\{x\mid \mu_i(x)=must\mhyphen out\} \cap \{x\mid \mu_i(x) \in \{in,out,blank\}\}=\emptyset, \text{ and } \\
\{x\mid \mu_i(x)=must\mhyphen out\} \cup \{x\mid \mu_i(x) \in \{in,out,blank\}\}=A,
\end{array}\end{equation}
which immediately demonstrates (16.4). Thus, the proof of (16.2) is complete.\placeqed	
\end{proof}

\numberwithin{equation}{proposition}
\begin{proposition}
Let $H=(A,R)$ be an {\sc af} and $\pi: A \rightarrow \{0,1,2,...,|A|\}$ be a total mapping such that for all $x \in A$ it holds that $\pi_1(x)=|\{x\}^- \setminus \{y: (y,y) \in R\}|$. And let $\mu: A \rightarrow \{in,out,must\mhyphen out,blank\}$ be a total mapping such that for all $x \in A$,
\[
\mu_1(x)=
\begin{cases}
must\mhyphen out, &~if~(x,x)\in R;\\
blank, &~otherwise.\\
\end{cases}
\]
Assume that Algorithm 2 is started with list-stb-ext($\mu_1$, $\pi_1$, $\{x \mid \mu_1(x)=blank \wedge \pi_1(x)=0\}$). Then, the algorithm computes exactly the stable extensions in $H$.
\end{proposition}
\begin{proof}\renewcommand\qedsymbol{}
	The proof is composed of two parts:
	\begin{enumerate}	
		\item[\emph{P1}.] At some state $i$, let $\{x\mid \mu_i(x)=in\}$ be the set reported by Algorithm~2 at line 13. Then, we need to prove that $\{x\mid \mu_i(x)=in\}$ is a stable extension in $H$. 
		\item[\emph{P2}.] For all $Q$, if $Q$ is a stable extension in $H$ and Algorithm 2 is sound (i.e. \emph{P1} is established), then there is a stable extension $\{x\mid \mu_i(x)=in\}$, reported by the algorithm at line 13 at some state $i$, such that $\{x\mid \mu_i(x)=in\}=Q$.		
	\end{enumerate}
	To establish \emph{P1}, it suffices to show that $\{x\mid \mu_i(x)=in\}=A \setminus \{x\mid \mu_i(x)=in\}^+$, which means it is required to prove that
	\begin{equation}
	\{x\mid \mu_i(x)=in\} \cap \{x\mid \mu_i(x)=in\}^+ = \emptyset,
	\end{equation}
	and
	\begin{equation}
	\{x\mid \mu_i(x)=in\} \cup \{x\mid \mu_i(x)=in\}^+=A.
	\end{equation}
	Observe, (17.1) is already established in Proposition 14. Now, we prove (17.2). As $\mu: A\rightarrow \{in,out,must\mhyphen out,blank\}$ is a total mapping,
	\begin{equation}
	\begin{array}{l}
\{x\mid \mu_i(x)=in\} \cup \{x\mid \mu_i(x)=out\}~\cup \\
 \{x\mid \mu_i(x)=blank\} \cup \{x\mid \mu_i(x)=must\mhyphen out\}=A.
\end{array}
	\end{equation}
	According to the algorithm's actions (lines 2--5), 
	\[\forall x \text{ with } \mu_i(x)=in, \forall y \in \{x\}^+, \mu_i(y)=out\]
	and
	\[\forall y \text{ with }\mu_i(y)=out, \exists x \text{ with } \mu_i(x)=in \text{ such that }y \in \{x\}^+.\]
	Thus, for every state $i$ it holds that
	\begin{equation}
	\{x\mid \mu_i(x)=in\}^+=\{x\mid \mu_i(x)=out\}.
	\end{equation}
	Observe, Algorithm 2 reports that $\{x\mid \mu_i(x)=in\}$ is stable if and only if $\{x\mid \mu_i(x)=blank\}=\emptyset$ and $\{x\mid \mu_i(x)=must\mhyphen out\}=\emptyset$, see lines 12--13. Using (17.4), we rewrite (17.3) as
	\begin{equation}
	\{x\mid \mu_i(x)=in\} \cup \{x\mid \mu_i(x)=in\}^+ \cup \emptyset \cup \emptyset=A,
	\end{equation}
	which means that (17.2) holds.	\\
Now we prove \emph{P2}. We rewrite \emph{P2} (by modifying the consequence) into \emph{\'P2}.\\ \\
\emph{\'P2}: For all $Q$, if $Q$ is a stable extension in $H$ and Algorithm 2 is sound, then there is a stable extension $\{x\mid \mu_i(x)=in\}$, reported by the algorithm at line 13 at some state $i$, such that for all $a \in Q$ it holds that $\mu_i(a)=in$. \\ \\
We establish \emph{\'P2} by contradiction. Later, we show that the consequence of \emph{\'P2} is equivalent to the consequence of \emph{P2}. Now, assume that \emph{\'P2} is false. \\ \\\
\emph{Negation of \'P2}: There is $Q$ such that $Q$ is a stable extension in $H$, Algorithm~2 is sound, and for every $\{x\mid \mu_i(x)=in\}$ reported by the algorithm at line 13 at some state~$i$, there is $a \in Q$ such that $\mu_i(a) \neq in$.\\ \\
We identify four cases.
\paragraph{Case 1} For $\mu_1(a)=blank$, if the algorithm terminates at line 8 during the very first execution of the while block (but not necessarily from the first round), then, since the algorithm is sound, $H$ has no stable extensions. This contradicts the assumption that $Q \supseteq \{a\}$ is a stable extension in~$H$. Hence, \emph{\'P2} holds.
\paragraph{Case 2} If $(a,a) \in R$, then $\mu_1(a)=must\_out$, which contradicts the assumption that $Q\supseteq \{a\}$ is a stable extension in $H$. Hence, \emph{\'P2} holds. 
\paragraph{Case 3} With $\mu_1(a)=blank$, assume that after the very first execution of the while block (i.e. including one or more rounds), $\mu_k(a)= blank$ for some state $k\ge1$. Then, for a state $i\ge k$, let $x=a$ (see line 15 in Algorithm 2). If for all subsequent states $j > i$, the set $\{x\mid \mu_j(x)=in\} \supseteq \{a\}$ is not reported stable by the algorithm, then, since the algorithm is sound, $a$ does not belong to any stable extension. This contradicts the assumption that $Q \supseteq \{a\}$ is a stable extension in $H$. Hence, \emph{\'P2} holds. 
\paragraph{Case 4} With $\mu_1(a)=blank$, assume that after the very first execution of the while block (i.e. including one or more rounds), $\mu_i(a) \neq blank$ for some state $i>1$. This implies, according to the while block's actions, that $\mu_i(a) \in \{in,out,must\_out\}$. For $\mu_i(a)=in$, if for all subsequent states $j > i$, the set $\{x\mid \mu_j(x)=in\} \supseteq \{a\}$ is not reported stable by the algorithm, then, since the algorithm is sound, $a$ does not belong to any stable extension. This contradicts the assumption that $Q \supseteq \{a\}$ is a stable extension in $H$. Likewise, for $\mu_i(a)\in\{out,must\_out\}$, since the algorithm's actions are sound, this implies that $a$ does not belong to any stable extension. This contradicts the assumption that $Q \supseteq \{a\}$ is a stable extension in $H$. Hence, \emph{\'P2} holds.\\ \\
Now we rewrite the consequence of \emph{\'P2}.\\ \\
\emph{The consequence of \'P2}: There is a stable extension $\{x\mid \mu_i(x)=in\}$, reported by the algorithm at line 13 at some state $i$, such that $Q \subseteq \{x\mid \mu_i(x)=in\}$.\\ \\
$Q$ being a proper subset of $\{x\mid \mu_i(x)=in\}$ is impossible because otherwise it contradicts that the algorithm is sound or that $Q$ is stable. Therefore, the consequence of \emph{\'P2} can be rewritten as next.\\ \\
\emph{The consequence of \'P2}: There is a stable extension $\{x\mid \mu_i(x)=in\}$, reported by the algorithm at line 13 at some state $i$, such that $Q = \{x\mid \mu_i(x)=in\}$.\\ \\
This is exactly the consequence of \emph{P2}. \placeqed
\end{proof}
\section{Conclusion}
We presented formal validation of a known labeling algorithm for listing all stable extensions in a given abstract argumentation framework. Our validation process given in this paper encourages more investigations in the research arena of abstract argumentation. Despite being experimentally verified, several existing labeling algorithms for abstract argumentation lack formal validation, which can be done in the spirit of this article. Moreover, since the core of labeling algorithms is backtracking, this paper might stimulate further work to validate backtracking procedures in some other application domains.

\bibliographystyle{plain}
\bibliography{references}

\end{document}